\documentclass[10pt]{IEEEtran}
\usepackage{amsmath,amsthm,amssymb}
\usepackage{graphicx,epstopdf,epsfig}
\usepackage{float}
\usepackage{xcolor}	
\usepackage{url}
\usepackage[caption=false,font=footnotesize]{subfig}

\newtheorem{theorem}{Theorem}
\newtheorem{lemma}[theorem]{Lemma}
\newtheorem{proposition}[theorem]{Proposition}
\newtheorem{assumption}[theorem]{Assumption}
\newtheorem{definition}[theorem]{Definition}

\allowdisplaybreaks

\def\x{\boldsymbol{x}}
\def\u{\boldsymbol{u}}
\def\v{\boldsymbol{v}}
\def\D{\mathcal{D}}
\def\TH{\boldsymbol{\theta}}
\def\lb{[\![}
\def\rb{]\!]}

\begin{document}

\title{On the Proof of Fixed-Point Convergence for Plug-and-Play ADMM}
\author{Ruturaj~G.~Gavaskar and Kunal~N.~Chaudhury,~\IEEEmembership{Senior~Member,~IEEE}}

\maketitle

\begin{abstract}
In most state-of-the-art image restoration methods, the sum of a data-fidelity and a regularization term is optimized using an iterative algorithm such as ADMM (alternating direction method of multipliers). In recent years, the possibility of using denoisers for regularization has been explored in several works. A popular approach is to formally replace the proximal operator within the ADMM framework with some powerful denoiser. However, since most state-of-the-art denoisers cannot be posed as a proximal operator, one cannot guarantee the convergence of these so-called plug-and-play (PnP) algorithms. In fact, the theoretical convergence of PnP algorithms is an active research topic. In this letter, we consider the result of Chan et al. (IEEE TCI, 2017), where fixed-point convergence of an ADMM-based PnP algorithm was established for a class of denoisers. We argue that the original proof is incomplete, since convergence is not analyzed for one of the three possible cases outlined in the paper. Moreover, we explain why the argument for the other cases does not apply in this case. We give a different analysis to fill this gap, which firmly establishes the original convergence theorem.
\end{abstract}


\section{Introduction}
\label{sec:intro}

A variety of image restoration problems, such as superresolution, deblurring, compressed sensing, tomography etc., are modeled as optimization problems of the form
\begin{equation}
\label{eq:main_prob}
\min_{\x \in \mathbb{R}^d} \quad f(\x) + \lambda g(\x),
\end{equation}
where the data-fidelity term $f(\x)$ is derived from the degradation and noise models, while the regularizer $g(\x)$ is derived from some prior on the ground-truth image \cite{Gunturk2012_img_restoration}. Traditionally, the regularizer is a sparsity-promoting function in some transform domain \cite{Elad2007_priors}.
In recent years, researchers have explored the possibility of using powerful Gaussian denoisers such as NLM \cite{Buades2005_NLM}  and BM3D \cite{Dabov2007_BM3D} for regularization purpose.
In \cite{Romano2017_RED,Reehorst2019_RED_clarifications}, the regularizer is explicitly constructed from a denoiser. 
On the other hand, for plug-and-play (PnP) methods  \cite{Venkatakrishnan2013_PnP,Sreehari2016_PnP,Ono2017_PD_PnP,Kamilov2017_PnP_ISTA}, the denoiser is formally substituted in place of the proximal operator in iterative algorithms such as FISTA \cite{Beck2009_FISTA}, primal-dual splitting \cite{Chambolle2011_PD}, and ADMM \cite{Boyd2011_ADMM}. 

The focus of this work is on an ADMM-based PnP method \cite{Chan2017_PnP_fixed_point}. We recall that the ADMM based solution of \eqref{eq:main_prob} involves the following steps  \cite{Boyd2011_ADMM}:
\begin{align}
\x_{k+1} &= \arg \min_{\x} \ f(\x) + (\rho/2) \lVert \x - (\v_k - \u_k) \rVert^2, \label{eq:ADMM_x_update} \\
\v_{k+1} &= \arg \min_{\v} \ \frac{\rho}{2 \lambda} \lVert \v - (\x_{k+1} + \u_k) \rVert^2 + g(\v), \label{eq:ADMM_v_update} \\
\u_{k+1} &= \u_k + \x_{k+1} - \v_{k+1} \label{eq:ADMM_u_update},
\end{align}
where $\rho>0$ is a penalty parameter and $\lVert \cdot \rVert$ is the Euclidean norm (this is the rescaled form of ADMM).
If $f$ and $g$ are convex, then under some technical conditions, the iterates $(\x_k)$ are guaranteed to converge to a fixed-point, which is the global minimizer of \eqref{eq:main_prob}.
Now, \eqref{eq:ADMM_v_update} corresponds to regularized Gaussian denoising, where $g$ assumes the role of the regularizer \cite{Hunt1977_bayesian_methods}. Based on this observation, the original proposal in 
\cite{Venkatakrishnan2013_PnP} was to replace the proximal operation \eqref{eq:ADMM_v_update} with an off-the-shelf denoiser, i.e., the $\v$-update is replaced by $\v_{k+1} = \D_{\sigma}(\x_{k+1} + \u_k)$, where $\D_{\sigma}$ is the denoiser in question. The idea is simply to exploit the excellent denoising capability of state-of-the-art denoisers for restoration, even though we might not be able to conceive  them as proximal operators (of some regularizer). We refer the readers to \cite{Venkatakrishnan2013_PnP,Sreehari2016_PnP} for  a detailed account. The technical challenge, however, is that the resulting sequence of operations, referred to as PnP-ADMM, need not necessarily correspond to an optimization problem. As a result, the convergence of the iterates is at stake. In particular, we can no longer relate to the optimization in \eqref{eq:main_prob} and use existing results \cite{Boyd2011_ADMM} to ensure convergence. Nevertheless, PnP-ADMM is often found to converge empirically and yields high-quality reconstructions in several applications \cite{Venkatakrishnan2013_PnP,Sreehari2016_PnP,Chan2017_PnP_fixed_point,Teodoro2019_PnP_fusion}.
Among other things, questions relating to the convergence and optimality of PnP-type methods have been studied in recent works.
In \cite{Sreehari2016_PnP}, convergence guarantees were derived for a kernel-based denoiser for PnP-ADMM.
Later, it was shown in \cite{Chan2017_PnP_fixed_point} that the convergence can be ensured for a broad class of denoisers. Apart from ADMM, PnP algorithms based on various iterative methods have been explored in \cite{Ono2017_PD_PnP,Kamilov2017_PnP_ISTA,Sun2019_PnP_SGD,Meinhardt2017_learning_prox_op,Ryu2019_PnP_trained_conv,Dong2018_DNN_prior}.
We note that denoisers have also been used for regularization purpose in \cite{Brifman2016_PnP_superresolution,Teodoro2019_targeted_PnP,Tirer2019_iter_denoising,Tirer2019_superres_CNN_denoiser,Fletcher2018_PnP_analysis,Yazaki2019_interp_graph_PnP,Chen2018_superres_PnP}.
The relation of PnP-ADMM with graph Laplacian-based regularization was investigated in \cite{Chan2019_PnP_graph_SP}, whereas in \cite{Buzzard2018_PnP_CE} a framework motivated by PnP, called Consensus Equilibrium, was proposed.

In this letter, we revisit the proof of convergence of the PnP-ADMM algorithm in \cite{Chan2017_PnP_fixed_point} and address an inadequacy therein. It was proved that, under suitable assumptions, the sequence of iterates $(\x_k,\v_k,\u_k)_{k \geq 1}$ generated by this algorithm converges to a fixed-point, for any arbitrary initialization $(\x_0,\v_0,\u_0)$. Instead of a fixed $\rho$, an adaptive $\rho_k$ is used in \cite{Chan2017_PnP_fixed_point}, which
plays an important role in the proof. However, this necessitates the use of a case-by-case approach conditioned on the adaptation rule (see Section \ref{sec:background} for details). Of the three cases considered in the paper, convergence was proved for the first two cases. It was claimed that convergence for the third case automatically follows from that of the first two cases. However, we argue that this is generally not true and
hence a separate proof is needed for the third case. We give such a proof, which differs from the proof for the first two cases in \cite{Chan2017_PnP_fixed_point}. In particular, we show that the  difference between successive iterates is bounded by a \textit{piecewise geometric sequence}, as opposed to a geometric sequence for the first two cases.
We prove that this sequence is summable, which is used to show that the iterates $(\x_k,\v_k,\u_k)_{k \geq 1}$  form a Cauchy sequence (and is hence convergent). 

The remaining letter is as follows.
In Section \ref{sec:background}, we review the algorithm in \cite{Chan2017_PnP_fixed_point}. In particular, we discuss the convergence result and explain why the proof is incomplete. The outline of the proof is provided in Section \ref{sec:results}, while 
the technical details are deferred to Section \ref{sec:proofs}.

\section{Background}
\label{sec:background}

As mentioned earlier, the updates in PnP-ADMM \cite{Chan2017_PnP_fixed_point} are modeled on the ADMM updates (\ref{eq:ADMM_x_update})--(\ref{eq:ADMM_u_update}), with the following changes: a denoiser is used in the $\v$-update, and  $\rho$ is updated in each iteration. In particular, the updates are given by
\begin{align}
\x_{k+1} &= \arg \min_{\x} \ f(\x) + \frac{\rho_k}{2} \lVert \x - \v_k + \u_k \rVert^2, \label{eq:x_update}\\
\v_{k+1} &= \D_{\sigma_k} (\x_{k+1} + \u_k), \label{eq:v_update}\\
\u_{k+1} &= \u_k + \x_{k+1} - \v_{k+1}, \label{eq:u_update}
\end{align}
where $\sigma_k = \sqrt{\lambda/\rho_k}$. Here $\D_{\sigma} : \mathbb{R}^d \to \mathbb{R}^d$ is a denoising operator,  where the parameter $\sigma$ controls its denoising action. It was proposed to update $\rho_k$ based on the residual
\begin{equation}
\label{Dk}
\Delta_{k+1} = D(\TH_k,\TH_{k+1}),
\end{equation}
where the metric $D$ on $\mathbb{R}^d \times \mathbb{R}^d \times \mathbb{R}^d$ is defined as
\begin{equation*}
D(\TH_1,\TH_2) = \frac{1}{\sqrt{d}} \left( \lVert \x_{1} - \x_{2} \rVert + \lVert \v_{1} - \v_{2} \rVert + \lVert \u_{1} - \u_{2} \rVert \right);
\end{equation*}
$\TH_1=(\x_1,\v_1,\u_1), \TH_2=(\x_2,\v_2,\u_2)$, and the three components of $\TH_1$ and $\TH_2$ are vectors in $\mathbb{R}^d$. Thus, \eqref{Dk}  is simply the distance between the $k$-th and $(k+1)$-th iterates, which measures the progress made by the algorithm. The exact rule proposed in \cite{Chan2017_PnP_fixed_point}  is as follows:
\begin{equation}
\label{eq:rho_update}
\rho_{k+1} = 
\begin{cases}
\gamma \rho_k, & \text{ if } \Delta_{k+1} \geq \eta \Delta_k \quad (\text{condition } C_1) \\
\rho_k, & \text{ if } \Delta_{k+1} < \eta \Delta_k  \quad (\text{condition } C_2),
\end{cases}
\end{equation}
where $\gamma > 1$ and $0 < \eta < 1$ are predefined parameters. The above rule, in effect, decreases the denoising strength $\sigma$ ($\rho$ is increased) if the ratio of the current and previous residuals is greater than $\eta$; else, $\rho$ is kept unchanged (see \cite{Chan2017_PnP_fixed_point} for details).

It was claimed in \cite{Chan2017_PnP_fixed_point} that the iterates generated by (\ref{eq:x_update})--(\ref{eq:u_update}) converge to a fixed point if a couple of assumptions are met. The first concerns the data-fidelity term.
\begin{assumption}
\label{asm:bounded_grad}
The function $f$ is differentiable and there exists $M > 0$ such that $ \lVert \nabla \! f(\x) \rVert \leq M \sqrt{d}$ for all $\x$.
\end{assumption}

\noindent The second assumption concerns the denoiser.
\begin{assumption}
\label{asm:bounded_den}
There exists $K > 0$ such that, for all $\x$,
\begin{equation}
\label{eq:bounded_den}
\lVert \D_{\sigma}(\x) - \x \rVert^2 \leq K d \sigma^2.
\end{equation}
\end{assumption}
While discussions on the above assumptions can be found in \cite{Chan2017_PnP_fixed_point}, here we reiterate a couple of remarks about Assumption \ref{asm:bounded_den}. It is difficult to mathematically verify \eqref{eq:bounded_den} even for simple denoisers, let alone sophisticated ones such as BM3D. However, an implication of \eqref{eq:bounded_den} is that the denoiser acts like an identity map (idle filter) when  $\sigma$ is close to zero. It is reasonable to expect that any practical denoiser obeys this weaker condition. Moreover, while the denoiser might not  perfectly behave as an identity operator when $\sigma$ is close to zero, it is possible to artificially force this behavior.

We are now ready to state the convergence result in \cite{Chan2017_PnP_fixed_point}.
\begin{theorem}
\label{thm:convergence}
Under Assumptions \ref{asm:bounded_grad} and \ref{asm:bounded_den}, the iterates $\{\TH_k = (\x_k,\v_k,\u_k): k \geq 1\}$ converge to a fixed point. 
\end{theorem}
In particular, the iterates do not diverge or oscillate. 
We note that convergence of $(\TH_k)$ implies that $\Delta_k \to 0$ as $k \to \infty$.
However, the converse is generally not true, i.e., it is possible that $\Delta_k$ converges to $0$ but $(\TH_k)$ do not converge. The technical point is that $\Delta_k$ must vanish sufficiently fast to guarantee the convergence of $(\TH_k)$. This is used in \cite{Chan2017_PnP_fixed_point}  as well as the present analysis.

To set up the technical context, we briefly recall the arguments provided in \cite{Chan2017_PnP_fixed_point} in support of  Theorem \ref{thm:convergence}. First, Assumptions \ref{asm:bounded_grad} and \ref{asm:bounded_den} were used to obtain the following result; see \cite[Appendix B, Lemma 1]{Chan2017_PnP_fixed_point}.
\begin{lemma}
\label{lem:delta_bound}
If condition $C_1$ in \eqref{eq:rho_update} holds at iteration $k$, then $\Delta_{k+1} \leq c/\sqrt{\rho_k}$ for some $c >0$.
\end{lemma}
\noindent Now, note that exactly one of the following cases must hold:
\begin{description}
\item[$(S_1)$] Condition $C_2$ holds for finitely many $k$.
\item[$(S_2)$] Condition $C_1$ holds for finitely many $k$.
\item[$(S_3)$] Both $C_1$ and $C_2$ hold for infinitely many $k$.
\end{description}
In \cite{Chan2017_PnP_fixed_point}, convergence was established for $S_1$ and $S_2$ as follows.
Suppose $S_1$ is true, and let $n_1$ be the largest $k$ when $C_2$ holds, i.e., $C_1$ is true for $k \geq n_1+1$.
Then it follows from \eqref{eq:rho_update} that $\rho_k$ increases monotonically: $\rho_k = \gamma^{k-n_1} \rho_{n_1}$  for $k \geq n_1$.
Using Lemma \ref{lem:delta_bound}, we can thus conclude that
\begin{equation}
\label{ineq1}
\Delta_{k+1} \leq \frac{c}{\sqrt{\rho_{n_1}} \sqrt{\gamma}^{k-n_1}} \qquad (k \geq n_1).
\end{equation}
Similarly, for $S_2$, let $n_2-1$ be the largest $k$ when $C_1$ holds, so that $C_2$ holds for $k \geq n_2$. By recursively applying the condition in \eqref{eq:rho_update}, we then obtain
\begin{align}
\label{ineq2}
\Delta_{k+1} &< \eta^{k+1-n_2} \Delta_{n_2} \\
& \leq \frac{c}{\sqrt{\rho_{n_2-1}}} \eta^{k+1-n_2} \qquad (k \geq n_2),
\end{align}
where the second inequality follows from Lemma \ref{lem:delta_bound}.
In summary, for both $S_1$ and $S_2$, we can find a sufficiently large $n$ and $A > 0$ such that 
\begin{equation}
\label{eq:geom_upper_bound}
\Delta_{k+1} \leq A \beta^{k} \qquad (k \geq n),
\end{equation}
where $0 < \beta < 1$. Namely, the error between successive iterates is eventually upper-bounded by a decaying geometric sequence. Using the triangle inequality, the fast convergence of $(\Delta_k)$ can be used to show that the original sequence $(\TH_k)$ is Cauchy, and hence convergent (since the ambient space is complete). This establishes the convergence of $(\TH_k)$ for the first two cases.

It was stated in \cite{Chan2017_PnP_fixed_point} that $S_3$ is a ``union of $S_1$ and $S_2$'', and that convergence under $S_1$ and $S_2$ implies convergence for $S_3$.
However, this is not true simply because the proof sketched above is valid only if one of $C_1$ or $C_2$ occurs finitely many times---this naturally excludes the case where \textit{both} $C_1$ and $C_2$ occur infinitely often.
For example, consider the hypothetical situation in which $C_1$ occurs for every even $k$ and $C_2$ occurs for every odd $k$.
Clearly, the proof does not work in this case.

For further clarity, let us carefully examine the technique in  \cite{Chan2017_PnP_fixed_point} used to establish convergence for $S_1$ and $S_2$.
For $S_1$, the eventual bound on $(\Delta_k)$ was established using the fact that $\rho_k$ is monotonically increasing for $k \geq n_1$.
A similar bound for $S_2$ was derived using the second inequality in \eqref{eq:rho_update}, which holds for $k \geq n_2$.
Thus, in both \eqref{ineq1} and \eqref{ineq2}, the existence of a finite $n_1$ (or $n_2$) is vital because it allows us to ignore the first few terms of the sequence $(\Delta_k)$, and understand its behavior over the tail. In turn, this is possible because condition $C_1$ (or $C_2$) occurs only a finite number of times.
If both $C_1$ and $C_2$ occur infinitely often, we cannot find a finite $n$ beyond which a single inequality holds for $\Delta_k$.
This is precisely why the technique in \cite{Chan2017_PnP_fixed_point} is not applicable for $S_3$.
Before proceeding further, we note why it is important to prove convergence in the case $S_3$.
Theorem \ref{thm:convergence} assures us that the algorithm converges regardless of which of the three cases hold. 
Therefore its proof remains incomplete unless convergence is proved for all three cases (and in particular $S_3$).
Moreover, experiments suggest that $S_3$ is indeed likely to arise in certain practical scenarios.
We have reported some empirical observations for deblurring and superresolution experiments in the supplementary material to back this. In these experiments, we found that when $\eta$ is close to $1$, it is likely that $S_3$ holds, i.e., the algorithm keeps switching between conditions $C_1$ and $C_2$.

\section{Main Result}
\label{sec:results}

We will now establish the  convergence of $(\TH_k)$ for case $S_3$.
In particular, we will show that $(\Delta_k)$ can be bounded by a sequence which vanishes sufficiently fast to ensure that $(\TH_k)$ is Cauchy. Such a sequence is defined next.
\begin{definition}
\label{def:PGS}
A positive sequence $(y_k)_{k \geq 1}$ is said to be a piecewise geometric sequence (PGS) if there exists  $0 < \beta <1$ and indices $n_1 < n_2 < \cdots$ such that 
\begin{itemize}
\item  for $j \geq 1$, the terms $y_{n_j+1},\ldots,y_{n_{j+1}}$ are in geometric progression with rate $\beta$, i.e., for $k=n_j+1,\ldots,n_{j+1}$,
\begin{equation*}
y_k = y_{n_j+1} \beta^{k-n_j-1}.
\end{equation*}
\item the subsequence $(y_{n_j+1})_{j \geq 1}$ is in geometric progression with rate $\beta$, i.e., for $j \geq 2$, $$y_{n_j+1} = y_{n_1+1} \beta^{j-1}.$$
\end{itemize}
\end{definition}
In other words, a PGS can be divided into chunks that are in geometric progression (with identical rates).
Moreover, the subsequence consisting of the peaks (i.e., the first term in each chunk) is itself in geometric progression. A PGS has a sawtooth-like appearance (see Figure \ref{fig:scan}), and is slower to decay to zero compared to a geometric sequence having the same rate. It turns out that the sequence of residues can be bounded by a PGS for case $S_3$.
\begin{lemma}
\label{lem:PGS_upper_bound}
Let $(\Delta_k)_{k \geq 1}$ be the residuals for case $S_3$. Then there exists a PGS $(y_k)_{k \geq 1}$ such that $\Delta_k \leq y_k$ for all $k$.
\end{lemma}
This may be considered as an analogue of \eqref{eq:geom_upper_bound} for $S_3$.
To deduce that $(\TH_k)$ is Cauchy, it suffices to show that a PGS is summable.
\begin{figure}[t!]
\centering
\includegraphics[width=0.8\linewidth,keepaspectratio]{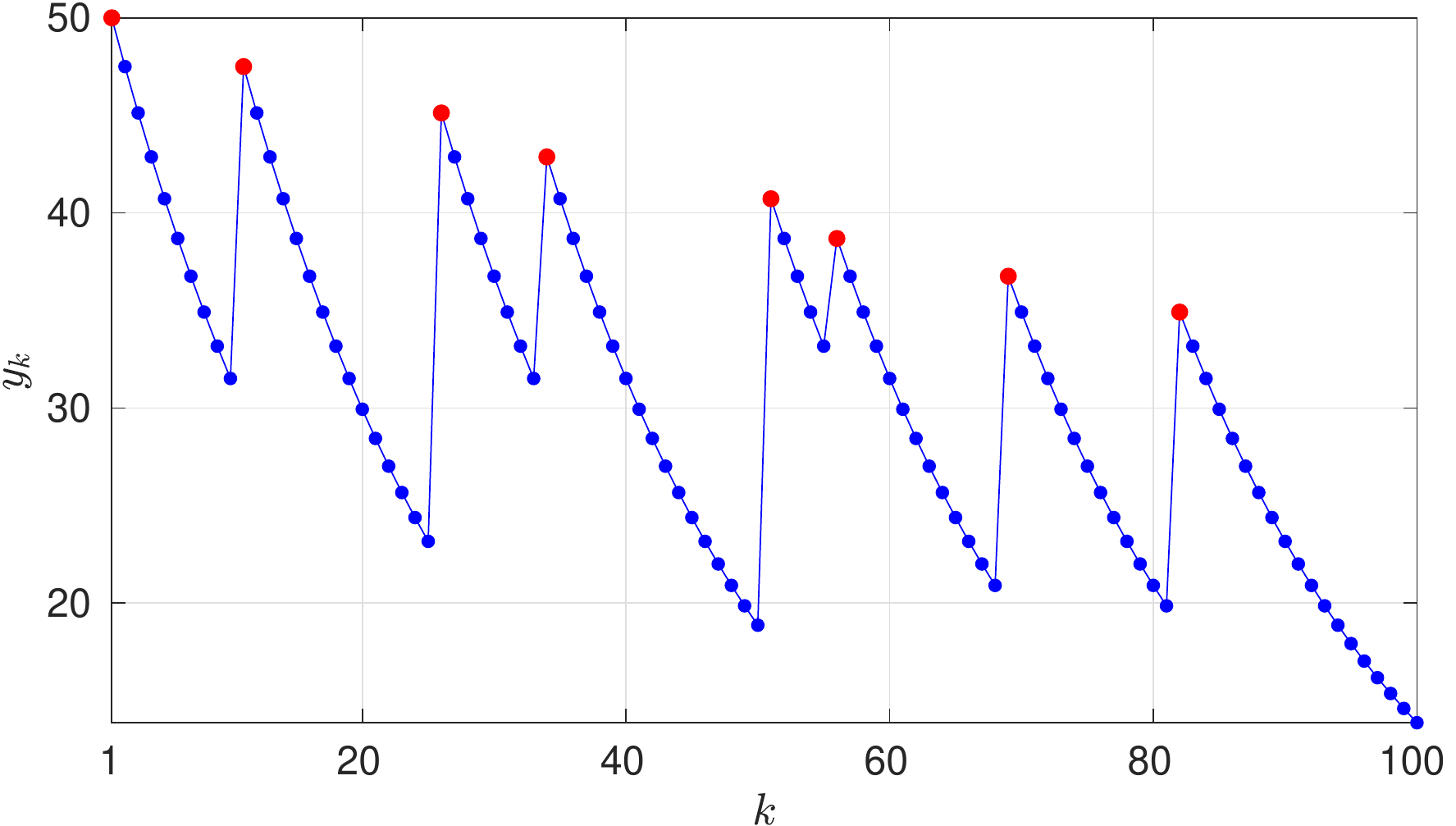}
\caption{A piecewise geometric sequence $(y_k)$. The red points are  the peaks, which themselves form a geometric sequence (cf. Definition \ref{def:PGS}).}
\label{fig:scan}
\end{figure}
\begin{lemma}
\label{lem:PGS_conv}
If $(y_k)_{k \geq 1}$ is a PGS, then $\sum_{k=1}^\infty y_k$ converges.
\end{lemma}
The proof of Lemma \ref{lem:PGS_upper_bound} and \ref{lem:PGS_conv} is somewhat technical and is deferred to Section \ref{sec:proofs}. Importantly, using the above lemmas, we can establish the convergence of $(\TH_k)$ for case $S_3$.
\begin{proposition}
\label{prop:convergence_S3}
The iterates $(\TH_k)_{k \geq 1}$ for case $S_3$ converge to a fixed point.
\end{proposition}
\begin{proof}
As noted earlier, all we need to show is that $(\TH_k)$ is a Cauchy sequence. That is, for any given $\epsilon >0$, we can find 
an integer $N$ such that $D(\TH_n,\TH_m) < \epsilon$ whenever $m > n \geq N$. Now, from the triangle inequality for metric $D$ and \eqref{Dk}, we have
\begin{align*}
D(\TH_n,\TH_m) & \leq D(\TH_n,\TH_{n+1}) + \cdots + D(\TH_{m-1},\TH_m)\\
&= \Delta_{n+1} + \cdots + \Delta_m.
\end{align*}
From Lemma \ref{lem:PGS_upper_bound} and \ref{lem:PGS_conv}, we can conclude that $ \Delta_1+\Delta_2+\cdots$ converges. This is because $(\Delta_k)$ is bounded by the PGS $(y_k)$, whose series itself converges. 
In particular, the partial sums of $(\Delta_k)$ form a Cauchy sequence. As a result, for any $\epsilon >0$, we can find 
a sufficiently large $N$ such that $\Delta_{n+1} + \cdots + \Delta_m < \epsilon$ when $m > n \geq N$. 
\end{proof}
Along with the proofs for cases $S_1$ and $S_2$ already provided in \cite{Chan2017_PnP_fixed_point}, Proposition \ref{prop:convergence_S3} completes the proof of Theorem \ref{thm:convergence}.

\section{Proofs}
\label{sec:proofs}
We now give the proofs for Lemma \ref{lem:PGS_upper_bound} and \ref{lem:PGS_conv}.
For integers $m > n$, we will use $\lb n,m \rb$ to denote the set $\{n,n+1,\ldots,m\}$.

\subsection{Proof of Lemma \ref{lem:PGS_upper_bound}}

Let $\alpha = 1/\sqrt{\gamma}$ and $\beta = \max \{\alpha,\eta\}$.
Note that $\beta \in (0,1)$.
We will show that $(\Delta_k)$ is bounded by a PGS with rate $\beta$.

Let $n_1$ be the iteration at which condition $C_1$ holds for the first time.
Further, let $m_1>n_1$ be the iteration at which condition $C_2$ occurs for the first time after $n_1$ (i.e. $C_1$ holds at iterations $n_1, n_1+1, \ldots,m_1-1$).
Let $n_2>m_1$ be the iteration at which $C_1$ holds for the first time after $m_2$, and so on.
Since $S_3$ holds, both $C_1$ and $C_2$ are true infinitely often. 
This gives us an infinite sequence of indices $n_1 < m_1 < n_2 < m_2 < \cdots$.
Now, by construction, for each $j \geq 1$, $C_1$ holds at iterations $k \in \lb n_j,m_j-1 \rb$. Hence, from \eqref{eq:rho_update}, $\rho_k = \gamma^{k - n_j} \rho_{n_j}$  for $k \in \lb n_{j+1},m_j \rb$.
Since this trivially also holds for $k=n_j$,
\begin{equation*}
\rho_k = \gamma^{k - n_j} \rho_{n_j}, \qquad k \in \lb n_j,m_j \rb.
\end{equation*}
By Lemma \ref{lem:delta_bound}, for $k \in \lb n_j+1,m_j \rb$, we have
\begin{equation*}
\Delta_k \leq c \rho_{k-1}^{-1/2} = c (\gamma^{k-n_j-1}\rho_{n_j})^{-1/2}.
\end{equation*}
Letting $L_j = c \rho_{n_j}^{-1/2}$, this becomes
\begin{equation}
\label{eq:bound_case1_old}
\Delta_k \leq L_j \alpha^{k-n_j-1}, \qquad k \in \lb n_j+1,m_j \rb.
\end{equation}

We now derive a relation between $L_j$'s for different $j$.
We know that $L_{j+1}/L_j = \sqrt{\rho_{n_j}/\rho_{n_{j+1}}}$.
However, from \eqref{eq:rho_update} we get,
\begin{equation*}
\rho_{n_{j+1}} = \rho_{n_{j+1}-1} = \cdots = \rho_{m_j} = \gamma^{m_j-n_j} \rho_{n_j},
\end{equation*}
since Case 2 occurs at iterations $m_j,\ldots,n_{j+1}-1$.
This gives
\begin{equation*}
L_{j+1} = L_j \alpha^{m_j-n_j} \leq L_j \alpha,
\end{equation*}
since $\alpha < 1$ and $m_j > n_j$.
Applying the above inequality recursively and using the fact $\alpha \leq \beta$, we get
\begin{equation*}
\label{eq:Lj}
L_j \leq L_1 \alpha^{j-1} \leq L_1 \beta^{j-1}.
\end{equation*}
Let $\bar{L}_j = L_1 \beta^{j-1}$.
Hence from \eqref{eq:bound_case1_old}, for $k \in \lb n_j+1,m_j \rb$,
\begin{equation}
\label{eq:bound_case1}
\Delta_k \leq \bar{L}_j \alpha^{k-n_j-1} \leq \bar{L}_j \beta^{k-n_j-1}.
\end{equation}
Now, condition $C_2$ holds for $k \in \lb m_j,n_{j+1}-1 \rb$.
Hence by recursively applying \eqref{eq:rho_update}, we obtain $\Delta_k \leq \eta^{k-m_j} \Delta_{m_j}$ for $k \in \lb m_j+1,n_{j+1} \rb$.
Note that this trivially also holds for $k=m_j$.
Hence, we have for $k \in \lb m_j,n_{j+1} \rb$,
\begin{align}
\Delta_k \leq \eta^{k-m_j} \Delta_{m_j} & \leq \beta^{k-m_j} \bar{L}_j \beta^{m_j-n_j-1} \nonumber \\
&= \bar{L}_j \beta^{k-n_j-1}, \label{eq:bound_case2}
\end{align}
where we have used \eqref{eq:bound_case1} with $k=m_j$ and the fact that $\eta \leq \beta$.

Combining \eqref{eq:bound_case1} and \eqref{eq:bound_case2} we get
\begin{equation*}
\Delta_k \leq \bar{L}_j \beta^{k-n_j-1}, \qquad k \in \lb n_j+1,n_{j+1} \rb.
\end{equation*}
In summary, we conclude that $(\Delta_k)$ is upper-bounded by the sequence $(y_k)$ defined by
\begin{equation*}
y_k = \bar{L}_j \beta^{k-n_j-1} = ( L_1 \beta^{j-1} ) \beta^{k-n_j-1},
\end{equation*}
for $j  \geq 1$ and $k \in \lb n_j+1,n_{j+1} \rb$. This does not specify the first $n_1$ terms of $(y_k)$; we may arbitrarily choose them to be equal to the corresponding terms of $(\Delta_k)$.
It follows from Definition \ref{def:PGS} that $(y_k)$ is indeed a PGS with rate $\beta$.

\subsection{Proof of Lemma \ref{lem:PGS_conv}}

Let the parameters $\beta$, $n_1 < n_2 < \cdots$ be as in Definition \ref{def:PGS}.
We will prove the convergence of $\sum_k y_k$ using the Cauchy criterion.
Let $\epsilon > 0$ be given.
We need to find an index $N$ such that 
\begin{equation}
\label{cauchy}
y_n +\cdots+y_m < \epsilon,  \qquad (m > n \geq N).
\end{equation}
Let $A = y_{n_1+1}$, and fix an integer $K > 0$ such that
\begin{equation}
\label{eq:bound2}
\beta^{K-1} < \frac{\epsilon (1-\beta)^2}{A}.
\end{equation}
This is possible since $\beta < 1$ and the right side of \eqref{eq:bound2} is positive.
We will prove that \eqref{cauchy} is satisfied by $N = n_K+1$.

First, for fixed $j \geq 1$, we derive a bound on the sum of the terms from $n_j+1$ to $n_{j+1}$. From Definition \ref{def:PGS}, we have
\begin{align}
y_{n_j+1} + \cdots + y_{n_{j+1}} &= y_{n_j+1} (1 + \beta + \cdots + \beta^{n_{j+1}-n_j-1}) \nonumber \\
&= y_{n_j+1} \cdot \frac{1 - \beta^{n_{j+1}-n_j}}{1-\beta} \nonumber \\
&< \frac{y_{n_j+1}}{1-\beta} = \frac{A \beta^{j-1}}{1-\beta}. \label{eq:bound1}
\end{align}
The inequality in the third step holds since $\beta^{n_{j+1}-n_j} < 1$, while the last equality follows from Definition \ref{def:PGS}.

We are now ready to establish \eqref{cauchy}.
Let $m,n$ be such that $m > n \geq n_K+1$.
Suppose $m$ lies in the chunk $\lb n_l+1,n_{l+1} \rb$ for some $l \geq K$.
Then $n_K+1 \leq n < m \leq n_{l+1}$.
As a result,
\begin{align*}
\sum_{k=n}^{m} y_k &\leq \sum_{k=n_K+1}^{n_{l+1}} y_k \\ 
&= \sum_{k=n_K+1}^{n_{K+1}} y_k + \cdots + \sum_{k=n_l+1}^{n_{l+1}} y_k \\
&< \frac{A}{1-\beta} \cdot \frac{\beta^{K-1}}{1-\beta}(1-\beta^{l-K+1}) \\
&< \frac{A \beta^{K-1}}{(1-\beta)^2} < \epsilon,
\end{align*}
where the inequality in the third step follows from \eqref{eq:bound1} and the last inequality follows from \eqref{eq:bound2}.
Therefore, $N = n_K+1$ satisfies the Cauchy criterion, where $K$ is defined by \eqref{eq:bound2}. This completes the proof.

\section{Conclusion}
\label{sec:conc}
We pointed out that the proof of convergence of the PnP-ADMM algorithm in \cite{Chan2017_PnP_fixed_point} does not address a certain case. We reasoned that this case needs to be handled differently from the cases addressed in \cite{Chan2017_PnP_fixed_point}. This is because the approach in \cite{Chan2017_PnP_fixed_point} fundamentally assumes that a certain condition holds a finite number of times, which is not true for the case in question.
In particular, we showed that unlike the geometric sequences used for the other cases, 
we need to work with a piecewise geometric sequence. Our proof of convergence follows from the observation that the residue between successive iterations is upper-bounded by this summable sequence.
Our analysis rigorously establishes the convergence theorem in \cite{Chan2017_PnP_fixed_point}.

We note that in practice, optimization algorithms, including PnP-ADMM, are terminated after a finite number of iterations. In particular, since the cases in the convergence analysis involve infinite number of iterations, which of these hold in practice cannot be ascertained empirically. Therefore, getting a guarantee on theoretical convergence has practical importance---it provides a mathematical justification to terminate the algorithm after a sufficiently many iterations.  This is precisely what was accomplished in this letter.

\bibliographystyle{IEEEtran}
\bibliography{citations}

\newpage
\onecolumn

\textbf{\Large Supplementary material}

\begin{figure}[H]
\setlength{\linewidth}{\textwidth}
\setlength{\hsize}{\textwidth}
\centering
\subfloat[$\eta = 0.1$.]{\includegraphics[width=0.30\linewidth,keepaspectratio]{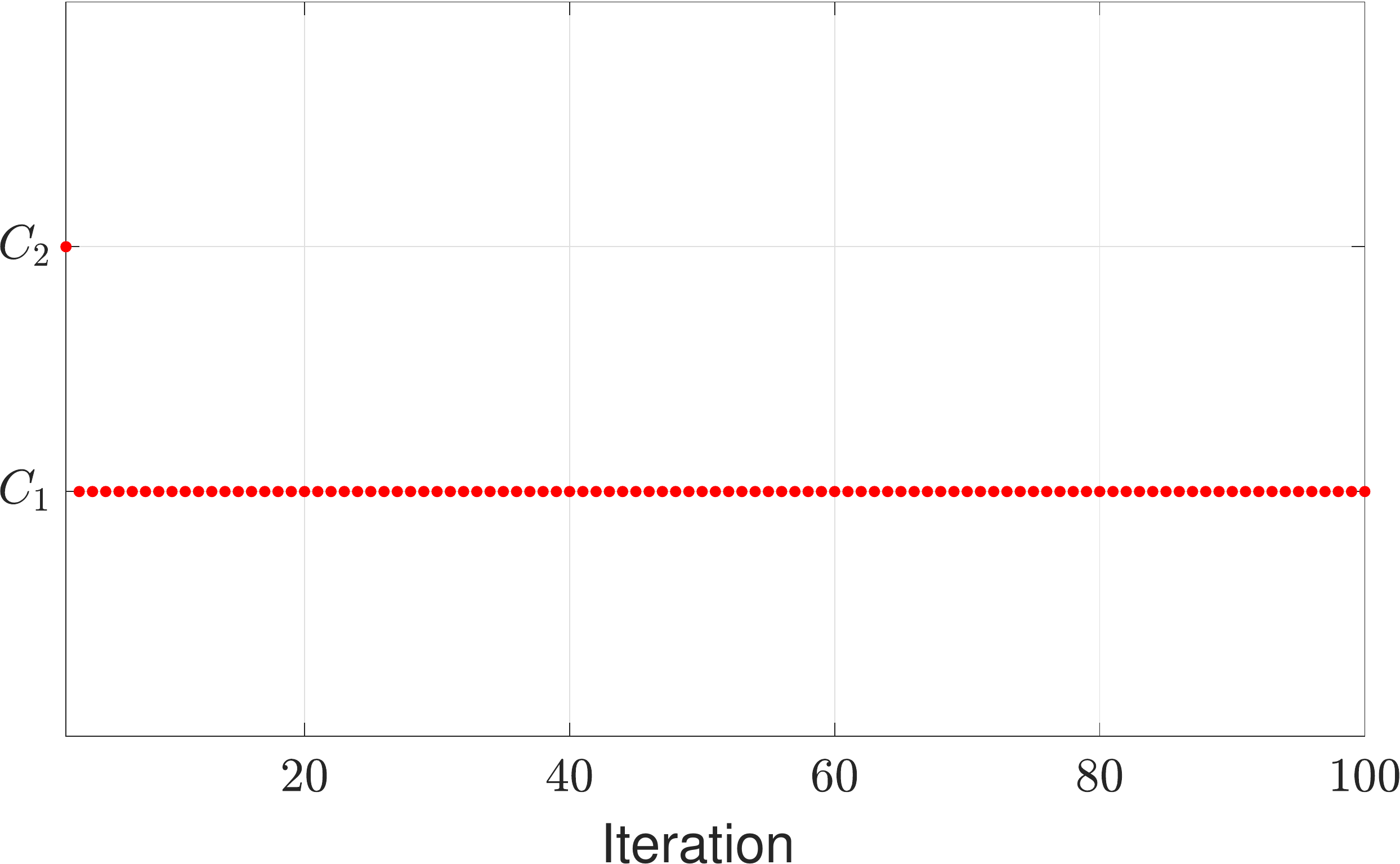}}
\hspace{0.1mm}
\subfloat[$\eta = 0.6$.]{\includegraphics[width=0.30\linewidth,keepaspectratio]{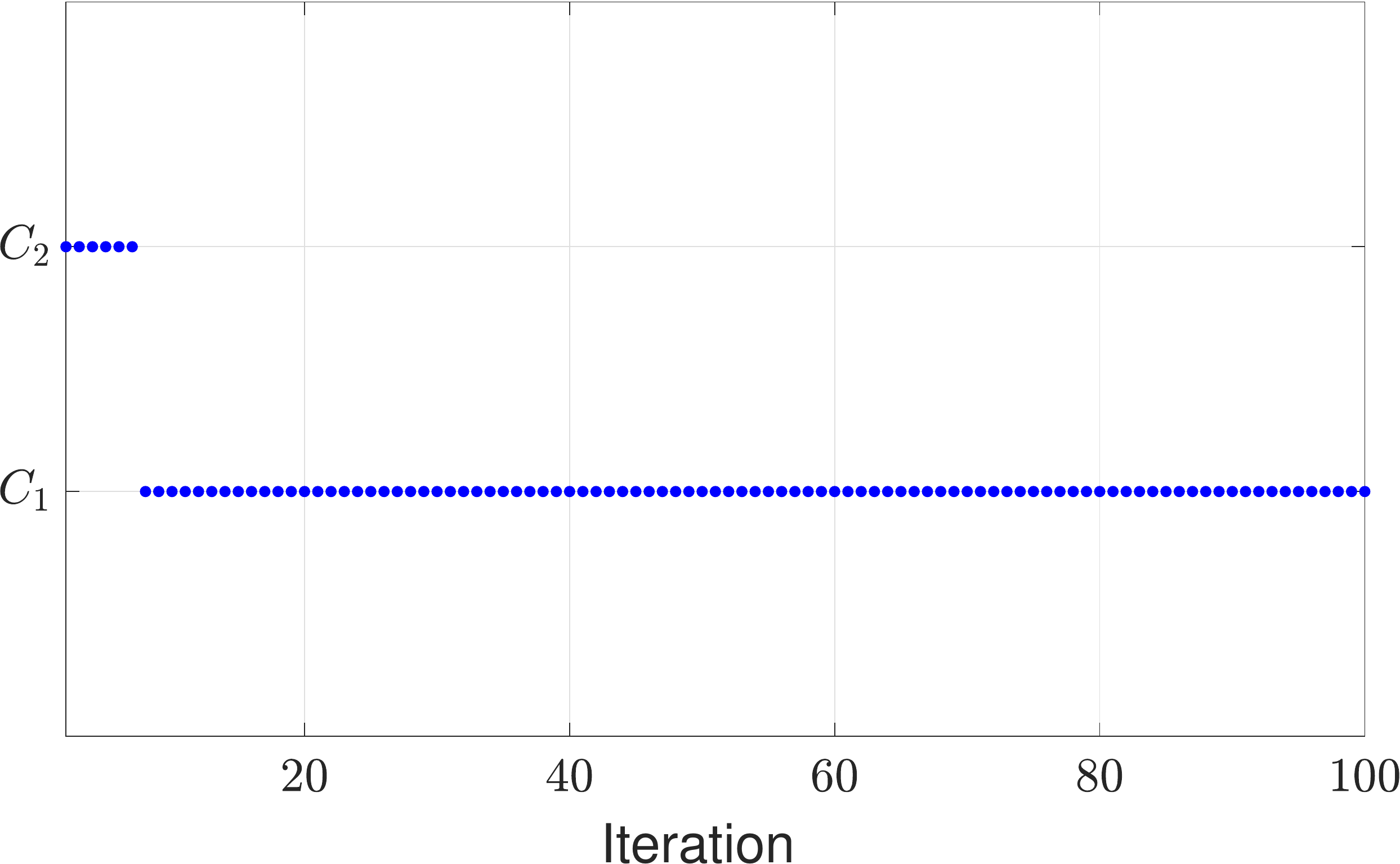}}
\hspace{0.1mm}
\subfloat[$\eta = 0.95$.]{\includegraphics[width=0.30\linewidth,keepaspectratio]{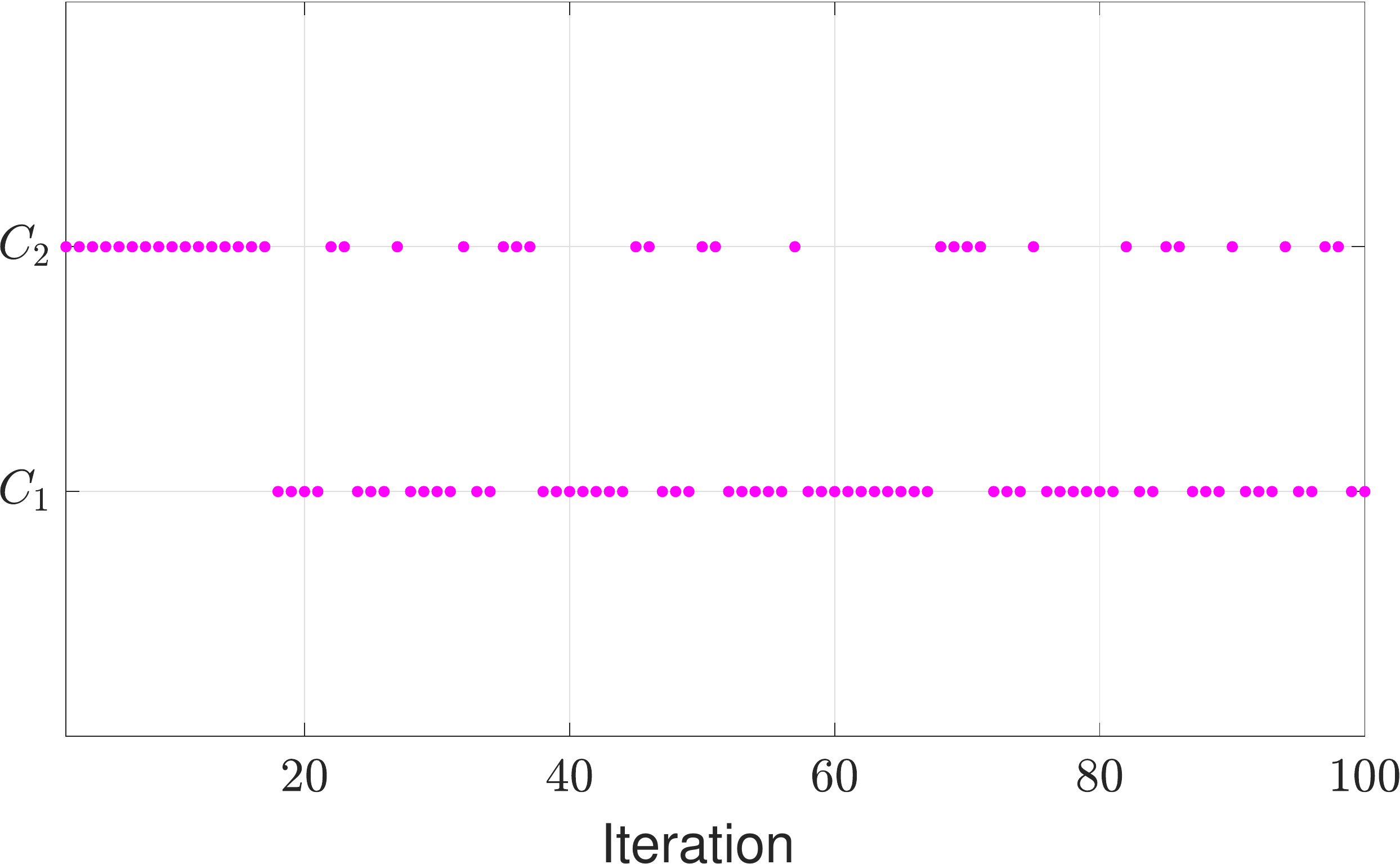}}\\
\subfloat[$\eta = 0.1$.]{\includegraphics[width=0.30\linewidth,keepaspectratio]{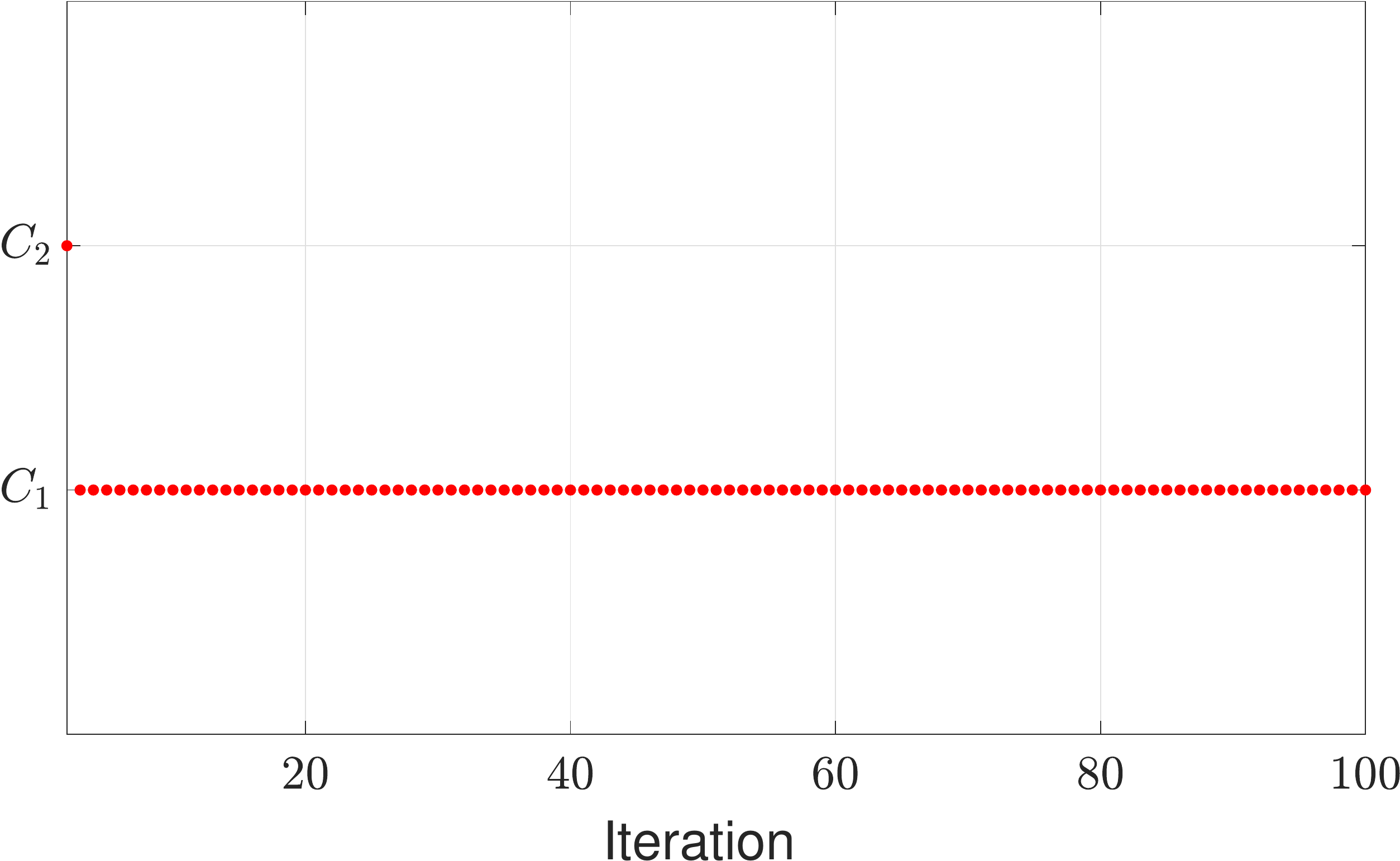}}
\hspace{0.1mm}
\subfloat[$\eta = 0.6$.]{\includegraphics[width=0.30\linewidth,keepaspectratio]{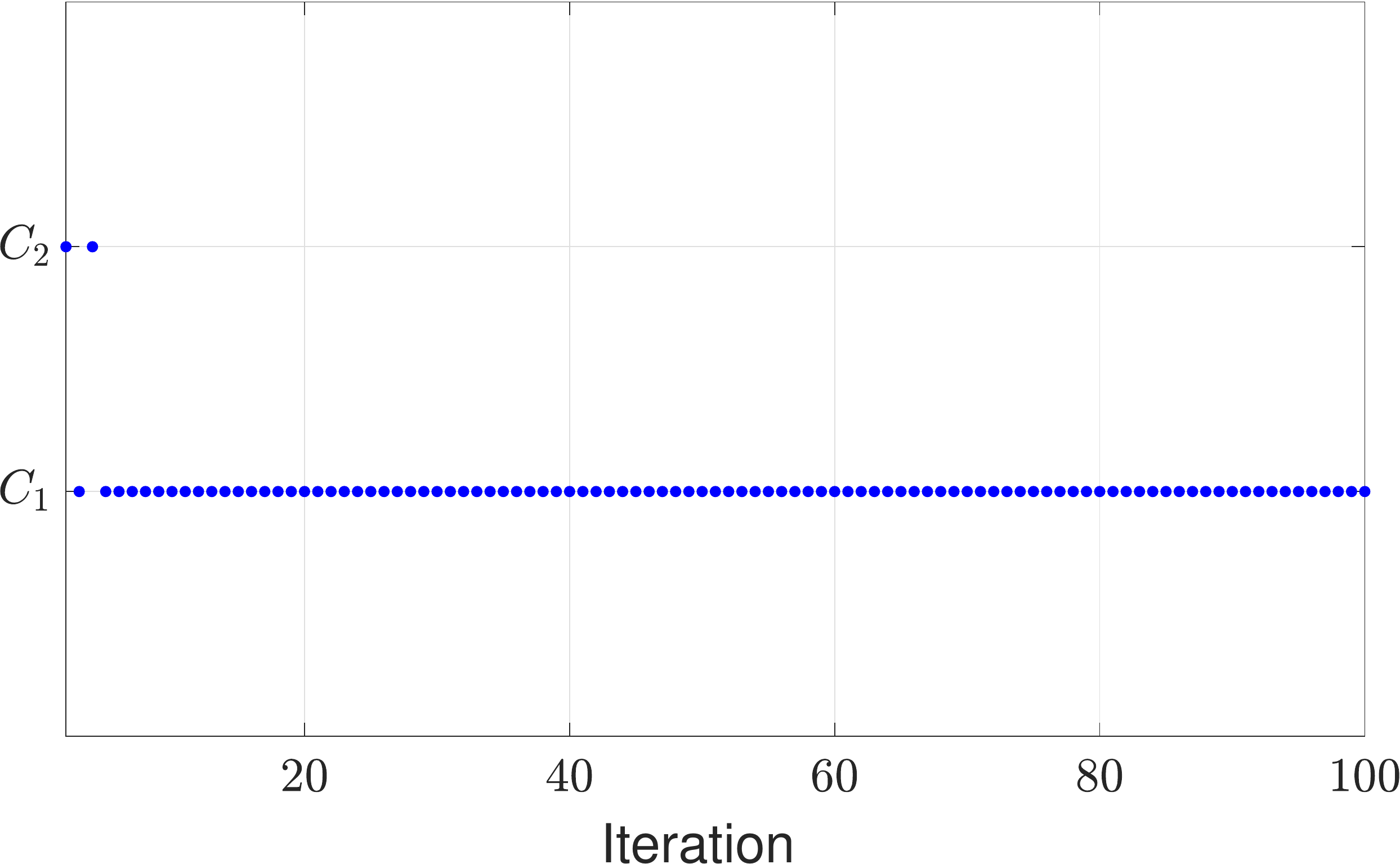}}
\hspace{0.1mm}
\subfloat[$\eta = 0.95$.]{\includegraphics[width=0.30\linewidth,keepaspectratio]{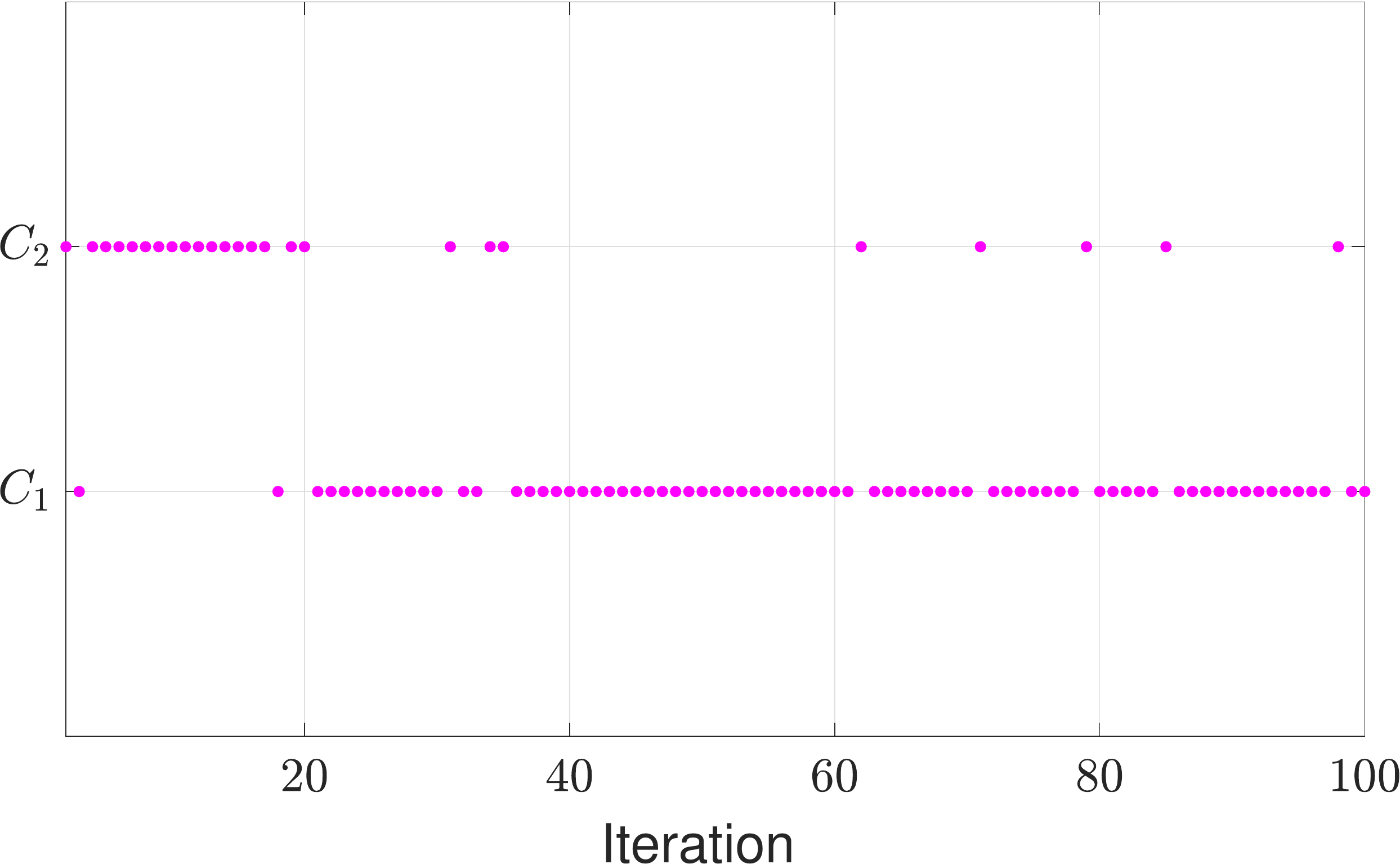}}
\caption{Iteration-wise check for conditions $C_1$ and $C_2$ for the PnP-ADMM algorithm (see main text for details). The top row is for a deblurring experiment and the bottom row for a superresolution experiment; the columns correspond to different settings of $\eta$. The number of iterations is $100$ in all cases.}
\label{fig:cond}
\end{figure}

In this supplement, we report some experimental observations on how the PnP-ADMM algorithm in \cite{Chan2016} switches between conditions $C_1$ and $C_2$. We have considered two image restoration problems, namely deblurring and superresolution. We have used the original Matlab code \cite{code} shared by the authors of \cite{Chan2016}.
We note that in the original code, the update $\rho_{k+1} = \gamma \rho_k$ is used at every iteration.
We have modified it as per the rule described in the manuscript.
The denoiser used in the code is BM3D.
All other settings are the default values used in the original code.
Also, the images used  for deblurring  and superresolution are \textit{house} and \textit{couple}, which are the default images used in the code.

Figure \ref{fig:cond} shows which of the two conditions ($C_1$ or $C_2$) holds for the first $100$ iterations.
This is done for $\eta = 0.1, 0.6$ and $0.95$.
It is observed that if $\eta$ is close to $1$ (rightmost plots), the algorithm tends to frequently switch between $C_1$ and $C_2$. This corresponds to case $S_3$ (provided the switching keeps occuring infinitely often), for which the convergence of PnP-ADMM  is studied in detail in the manuscript.
For smaller values of $\eta$, condition $C_2$ seems to occur only for the first few iterations, after which only $C_1$ holds. This corresponds to case $S_1$, for which the convergence was already studied in \cite{Chan2016}.

\bibliographystyle{IEEEbib}

\end{document}